%% file: BoissonJametMarcovici.tex
\documentclass[12pt,a4paper]{article}

\usepackage{a4,a4wide}

\usepackage[T1]{fontenc}
\usepackage{graphicx}
\usepackage{enumitem}

\usepackage[super]{nth}

\usepackage{algorithm}
\usepackage{algpseudocode}

\usepackage{xcolor}

\usepackage{hyperref}

\usepackage{amsmath, amssymb, amsthm}
\usepackage{yhmath}
\usepackage{stmaryrd}

\usepackage{stackengine}

\setstackgap{L}{1.2\normalbaselineskip}
\stackMath

\definecolor{codegreen}{rgb}{0,0.6,0}
\definecolor{codegray}{rgb}{0.5,0.5,0.5}
\definecolor{codepurple}{rgb}{0.58,0,0.82}
\definecolor{backcolour}{rgb}{0.95,0.95,0.92}

\usepackage{listings}
\lstdefinestyle{mystyle}{
	backgroundcolor=\color{backcolour},   
	commentstyle=\color{codegreen},
	keywordstyle=\color{magenta},
	numberstyle=\tiny\color{codegray},
	stringstyle=\color{codepurple},
	basicstyle=\ttfamily\footnotesize,
	breakatwhitespace=false,         
	breaklines=true,                 
	captionpos=b,                    
	keepspaces=true,                 
	numbers=left,                    
	numbersep=5pt,                  
	showspaces=false,                
	showstringspaces=false,
	showtabs=false,                  
	tabsize=2
}
\usepackage{tikz}
\usepackage{colortbl}
\usetikzlibrary{matrix,positioning,arrows.meta,arrows}
\usetikzlibrary{calc,matrix}

\usetikzlibrary{decorations.pathreplacing,matrix,calc,positioning}
\usepackage{nicematrix,tikz}

\def\N{\mathbb N}
\def\T{\mathbb T}
\def\X{\mathbb X}
\def\E{\mathbb E}
\def\P{\mathbb P}
\def\O{\mathcal O}

\def\ub#1{%
	\leavevmode\hbox{%
		\setbox0\hbox{#1}\dp0 0pt
		\vrule height.5ex width.4pt depth.33333ex \kern-.4pt
		\vtop{\hbox{\kern.15em \box0\kern.15em}\kern.33333ex \hrule}%
		\kern-.4pt \vrule height.5ex width.4pt depth.33333ex \kern-.4pt
	}%
}

\makeatletter
\def\underparen#1{\mathop{\vtop{\m@th\ialign{##\crcr
				$\hfil\displaystyle{#1}\hfil$\crcr
				\noalign{\kern3\p@\nointerlineskip}%
				\upparenfill\crcr\noalign{\kern3\p@}}}}\limits}
\def\upparenfill{$\m@th \setbox\z@\hbox{$\braceld$}%
	\bracelu\leaders\vrule \@height\ht\z@ \@depth\z@\hfill\braceru$}
\makeatother

\newcommand{\bbN}{\mathbb{N}}

\newtheorem{proposition}{Proposition}
\newtheorem{theorem}{Theorem}
\newtheorem{lemma}{Lemma}
\newtheorem{corollary}{Corollary}
\newtheorem{conjecture}{Conjecture}

\theoremstyle{definition}

	\title{On a probabilistic extension of the Oldenburger-Kolakoski sequence} %\thanks{...}\thanks{...}% At most 5 thanks
	
	\author{Chlo\'e Boisson\footnote{\'Ecole Normale Supérieure de Lyon, 15 parvis René Descartes, F-69342, Lyon, France.}, Damien Jamet\footnote{Université de Lorraine, Loria, UMR 7503, Vandœuvre-lès-Nancy, F-54506, France.} and Ir\`ene Marcovici\footnote{Université de Lorraine, CNRS, Inria, IECL, F-54000 Nancy, France.}
	\\ {\footnotesize \url{chloe.boisson@ens-lyon.fr}, \url{damien.jamet@loria.fr}, \url{irene.marcovici@univ-lorraine.fr}}}

\date{\today}
	
\begin{document}

\maketitle

\begin{abstract}
		The Oldenburger-Kolakoski sequence is the only infinite sequence over the alphabet $\{1,2\}$ that starts with 1 and is its own run-length encoding. 
		In the present work, we take a step back from this largely known and studied sequence by introducing some randomness in the choice of the letters written. This enables us to provide some results on the convergence of the density of $1$'s in the resulting sequence. When the choice of the letters is given by an infinite sequence of i.i.d. random variables or by a Markov chain, the average densities of letters converge. Moreover, in the case of i.i.d. random variables, we are able to prove that the densities even almost surely converge.  
	%\keywords{Oldenburger \and Kolakoski \and Combinatorics on words \and Random variables \and Markov chain}
	\end{abstract}
	%\subjclass{11K31, 11K36}
	
	\input{Section1.tex}

	\input{Section2.tex}

	\input{Section3.tex}
	\input{Section4.tex}

\input{Section5.tex}

\input{Section6.tex}

%\vspace{-0.5cm}

	\bibliography{mybiblio}{}
	\bibliographystyle{plain}
	
\end{document}

%% file: Section1.tex
	\section{Introduction}
	
	The Oldenburger-Kolakoski sequence $122112122122112\dots$ introduced by R. Oldenburger \cite{RO1939} and lately mentioned by W. Kolakoski \cite{WK1966} is the unique sequence  $x_1x_2x_3\dots$ over the alphabet $\{1,2\}$ with $x_1 = 1$ and whose $k$-th block has length $x_k$ for $k \in \bbN^\star$.
	
	In \cite{Keane91} M.S. Keane asked whether the density of $1$'s in this sequence is $1/2$. In \cite{chvatal93}, V.~Chv{\'a}tal showed that the upper density of $1$'s (resp. $2$'s) is less than 0.50084. This bound has been slightly improved by M. Rao but Keane's question still stands: << Is the density of $1$'s in Oldenburger-Kolakoski sequence defined and equal to $0.5$? >>
	
	By definition, the Oldenburger-Kolakoski sequence $\mathcal{O} = (x_n)_{n \in \bbN\star}$ is a fixed point of the run-length encoding operator denoted $\Delta$:
	
	\begin{eqnarray}\label{eq::block}
	\stackunder{\mathcal{O}}{\Delta(\mathcal{O})} &	\stackunder{=}{=}	&	 \underparen{1}_{1}  \underparen{22}_{2} \underparen{11}_{2} \underparen{2}_{1} \underparen{1}_{1} \underparen{22}_2 \underparen{1}_1 \underparen{22}_2 \underparen{11}_2 \dots = 1^1 2^2 1^2 2^1 1^1 2^2 1^1 2^2 1^2 \dots \\
	\mathcal{O} &	=	&	1^{x_1}2^{x_2}1^{x_3}2^{x_4}1^{x_5}2^{x_6}1^{x_7}2^{x_8}1^{x_9} = \prod_{n \in \bbN} (1^{x_{2n+1}}2^{x_{2n+2}})
	\end{eqnarray}

	In \cite{RO1939}, R. Oldenburger refers to sequences over an alphabet $\Sigma$ as \emph{trajectories} and refers to the	sequence $\Delta(w)$ as the \emph{exponent trajectory} of the trajectory $w$. He stated that << a periodic trajectory is distinct from its exponent trajectory >> (Theorem 2, \cite{RO1939}) and, therefore, the Oldenburger-Kolakoski sequence is not periodic.
		
	The Oldenburger-Kolakoski sequence is also connected to differentiable words, $C^\infty$-words and smooth words \cite{BBC05, BDLV06, Dek80}. %
	A sequence $w$ over the alphabet $\Sigma\subset \mathbb{N}^\star$ is \emph{differentiable} if and only if $\Delta(w)$ is also defined over the same alphabet $\Sigma$. The sequence $\Delta(w)$ is called the \emph{derivative sequence} of $w$~\cite{Dek80}. A \emph{$C^\infty$-word}, or \emph{smooth word}, is an infinitely differentiable sequence. Obviously, the Oldenburger-Kolakoski sequence is a $C^\infty$-word since it is a fixed-point of the run-length encoding operator $\Delta$.
	
%	The sequence $(x_n)_{n \in \mathbb{N}}$ is called the \emph{derivative sequence} of $\mathcal{O}$ \cite{Dek80}: $x_i$ is equal to the length of the $i^{th}$ run of $\mathcal{O}$.
	
	% The sequence $(12)^\omega$ is called the \emph{directing sequence} of $\mathcal{O}$ and we say that the sequence $\mathcal{O}$ is \emph{directed by} $(12)^\omega$.

	Although not answering Keane's question fully, F.M. Dekking established connections  between possible combinatorial properties of the Oldenburger-Kolakoski sequence \cite{Dek80}: if the Oldenburger-Kolakoski sequence is closed by complementation (that is, if $w$ occurs in $\mathcal{O}$ then so does $\widetilde{w}$ with $\widetilde{1}=2$ and $\widetilde{2}=1$) then it is recurrent (any word that occurs in $\mathcal{O}$ does so infinitely often) (Prop. 1, \cite{Dek80}). Moreover, the Oldenburger-Kolakoski sequence is closed by complementation if and only if it contains every finite $C^\infty$-word (Prop. 2, \cite{Dek80}).
	
A few years later, A. Carpi stated that the Oldenburger-Kolakoski sequence contains only a finite set of squares (words of the form $xx$ where $x$ is not empty) and does not contain any cube (word of the form $xxx$ where $x$ is not empty) \cite{Carpi93}. Hence, since $\mathcal{O}$ contains only squares of bounded length then it cannot be the fixed point of a non degenerated morphism: the image of a square $w=xx$ by such a morphism is still a square longer than $w$.

	There exist several ways to extend the definition of the Oldenburger-Kolakoski sequence, depending on whether one wants to preserve the fixed point property or to follow the construction scheme without requiring the resulting sequence to be a fixed point for the run-length encoding operator $\Delta$. For instance, one can deal with other alphabets and thus construct \emph{Generalized Oldenburger-Kolakoski} sequence (GOK-sequence for short) as follows: for any pair $(a,b)$ of non-zero natural numbers, there exists a unique fixed point $\mathcal{O}_{a,b}$ of $\Delta$ over the alphabet $\{a,b\}$ starting with $a$. %In other words, there exists a unique sequence $\mathcal{O}_{a,b}$  directed by $(ab)^\omega$ over the alphabet $\{a,b\}$. 
	Also, according to this notation, the original Oldenburger-Kolakoski sequence is $\mathcal{O}_{1,2}$.
	For instance, if $a=1$ and $b = 3$, the first terms of $\mathcal{O}_{1,3}$ are:
	\begin{equation}\label{eq::O_13}
	\mathcal{O}_{1,3} = \underparen{1}_{1} \underparen{333}_{3} \underparen{111}_{3} \underparen{333}_{3} \underparen{1}_1 \underparen{3}_1 \underparen{1}_1 \underparen{333}_3 \dots = 1^1 3^3 1^3 3^3 1^1 3^1 1^1 3^3\dots 
	\end{equation}

	A significant result is, unlike the case of the original Oldenburger-Kolakoski sequence, that the densities of $1$'s in $\mathcal{O}_{1,3}$ and $\mathcal{O}_{3,1}$ are known and approximately 0.3972  \cite{Sing04}.
	
Generalized Oldenburger-Kolakoski sequences are also connected with smooth words over arbitrary alphabets \cite{BertheBC05, BJP08}.
As for the (Generalized) Oldenburger-Kolakoski sequences, the properties of smooth words are better known for alphabets with letters of the same parity: for instance, while the frequency of letters in an infinite smooth word over $\{1,2\}$ is still unsolved, in \cite{BJP08} the authors showed that the frequency of letters for extremal smooth words (for the lexicographic order) over the alphabet $\{a,b\}$, where $a$ and $b$ are both even, is 0.5. They also computed the frequency for extremal smooth words over alphabets of type $\{1,b\}$, where $b$ is odd.
Moreover, if $a$ and $b$ have the same parity, then every infinite smooth word over the alphabet $\{a,b\}$ is recurrent \cite{BJP08}. Also, if $a$ and $b$ are both odd, then every infinite smooth word is closed under reversal but not under complementation \cite{BJP08}. On the other hand, if $a$ and $b$ are both even, then the extremal smooth words over the alphabet $\{a,b\}$ are neither closed under reversal nor closed under complementation \cite{BJP08}.

	For a more detailed survey on the Oldenburger-Kolakoski sequence and on generalizations over arbitrary two letter alphabets see \cite{Sing11}.

%% file: Section2.tex
\section{Extending the construction scheme to any directing sequence\label{sec:section2}}
		
\subsection{Notion of directing sequence.}
		
	In the construction scheme of a Generalized Oldenburger-Kolakoski sequence, the blocks of $\mathcal{O}_{a,b}$ are composed, alternatively, of $a$'s and $b$'s as shown in (\ref{eq::block}) when $a=1$ and $b=2$ and in (\ref{eq::O_13}) when $a=1$ and $b=3$. In other words, if $t_1=a$ and $t_2=b$, then <<~the $i^{\mathrm{th}}$ block of $\mathcal{O}_{a,b}=(x_i)_{i \in \mathbb{N}^\star}$ is of length $x_i$ and is filled with the letter $t_{i \mod 2}$~>>.
	 
	This construction scheme is clearly extendable to any finite sequences $T  = (t_1, t_2, \dots)$ over $\{a,b\}$ as follows: <<~the $i^{\mathrm{th}}$ block of $\mathcal{O}_T=(x_i)_{i \in \mathbb{N}^\star}$ is of length $x_i$ and is filled with the letter $t_i$~>> (see Program \ref{lst::liste1}). 

	\begin{minipage}{\linewidth}	
	\begin{lstlisting}[caption={\texttt{Python} function: $\mathcal{O}$ is an operates on sequences over $\mathbb{N}^\star$.},language=Python, label={lst::liste1}]
def O(T): 
	X = []
	k = 0
	for x in T:
		X += [x] # concatenate 'x' at the end of X
		X += [x]*(X[k]-1) # concatenate 'X[k]-1' copie(s) of x 
		k += 1
	return X
	\end{lstlisting}
	\end{minipage}

	We say that the sequence $\mathcal{O}_T$ is \emph{directed} by the sequence $T$ and the sequence $T$ is a \emph{directing sequence} of $\mathcal{O}_T$. 
	For instance, the sequence $\mathcal{O}_{a,b}$ is directed by $T=(ab)^\omega$ while $\mathcal{O}$ is directed by $(12)^\omega$.	Notice that the directed sequence $\mathcal{O}_T$ may no longer be a fixed point of the operator $\Delta$. 
	
		%Let $T=(t_n)_{n \in \bbN^\star}$ be a sequence over $\mathcal{A}$ and. 
		Let us now take a closer look at how the construction of  $\mathcal{O}_T$ provides a little more information than the sequence itself.  For instance, let $T=(t_n)_{n \in \bbN^\star} =21122\dots$ be a sequence over $\{1,2\}$:
	\begin{enumerate}[label={\textbf{Step \arabic*:}}]
		\item $\mathcal{O}_{t_1}=22$ and the second block is of length $2$: hence the \nth{3} and \nth{4} letters are in a same block of length 2. Let us denote $\mathcal{O}_{t_1} = \ub{22}\,\ub{??}$.
		\item $\mathcal{O}_{t_1 t_2}= \ub{22}\,\ub{11}\, \ub{?} \, \ub{?} $: the \nth{5} and the \nth{6} letter are respectively in blocks of length $1$.
		\item $ = \ub{22}\,\ub{11}\, \ub{1} \, \ub{?} \, \ub{?}$: the \nth{7} letter is in a block of length $1$.
	\end{enumerate}
	Roughly speaking, $t_1$ gives the length of all the blocks that contain up to the \nth{4} letter, $t_1t_2$ gives the length of all the blocks that contain up the \nth{6} letter and so on\dots~Let $w_n = \mathcal{O}_{t_1 \dots t_n}$ for each $n \in \bbN^\star$ , then 
	$\sum^{|w_1|}_{i = 1} { [w_1]_i} = 4$ , $\sum^{|w_2|}_{i = 1} { [w_2]_i} = 6$, $\sum^{|w_3|}_{i = 1} { [w_3]_i} = 7$\dots\\
	
\subsection{Partitions of the set of directing sequences}

We now introduce some subsets of directing sequences that will be crucial in the following sections. For this purpose, let us classify the sequences $T$ according to the information they provide on the length of the blocks of $\mathcal{O}_T$: let $k$ and $n$ be two integers such that $1 \leq k \leq n$ and let $\mathcal{S}_{n,k}$ be the set of sequences $T = (t_n)_{n \in \mathbb{N}}$ such that the length of the block of $\mathcal{O}_T$ containing its $n$\textsuperscript{th} letter is known when reading $t_k$ but not before. Formally, if $w_n = \mathcal{O}_{t_1 \dots t_n}$, then we have
	\begin{equation}\label{dfn::Snk}
	\mathcal{S}_{n,k} = \Big\{(t_1,...,t_n) \in \{1,2\}^n : \min \{j \in \llbracket 1;n \rrbracket, \sum^{|w_j|}_{i = 1} { [w_j]_i} \geq n \} = k\Big\}.
	\end{equation}
	
	Let us give a short example to illustrate the latter definition:  let $n = 5$ and let $T = (t_1,t_2,t_3,t_4,t_5) = (2,1,1,2,2)$. After the first step and the reading of $t_1$, we only know that $t_1=2$ and still do not know the length of the block that will contain the \nth{5} letter of $\mathcal{O}_T$. On the other hand, after having taken knowledge of the value of $t_2$, we know that the \nth{5} letter of $\mathcal{O}_T$ will be written in a block of size 1. Hence, $T \in \mathcal{S}_{5,2}$. More generally, 
	
	\begin{enumerate}[label={\textbf{Step \arabic*:}}]
	\item $\mathcal{O}_{t_1} = \ub{22}\,\ub{??}$ and $T \in \mathcal{S}_{3,1} \cap \mathcal{S}_{4,1}$ since $t_1$ provides the length of the block containing the \nth{3} and the \nth{4} of $\mathcal{O}_T$.
	\item $\mathcal{O}_{t_1 t_2} = \ub{22}\,\ub{11}\, \ub{?} \, \ub{?} $: the \nth{5} and the \nth{6} letter are respectively in blocks of length $1$. Hence $T \in \mathcal{S}_{5,2} \cap \mathcal{S}_{6,2}$.
	\item $\mathcal{O}_{t_1 t_2 t_3} = \ub{22}\,\ub{11}\, \ub{1} \, \ub{?} \, \ub{?}$ and  $T \in \mathcal{S}_{7,3}$.
	\item $\mathcal{O}_{t_1 t_2 t_3 t_4}  = \ub{22}\,\ub{11}\, \ub{1} \, \ub{2} \, \ub{?} \, \ub{??}$ and  $T \in \mathcal{S}_{8,4} \cap \mathcal{S}_{9,4}$.
	\item $\mathcal{O}_{t_1 t_2 t_3 t_4 t_5}  = \ub{22}\,\ub{11}\, \ub{1} \, \ub{2} \, \ub{2} \, \ub{??} \, \ub{??}$ and  $T \in \mathcal{S}_{10,5} \cap \mathcal{S}_{11,5}$.
	\end{enumerate}
	
%	The sets $\mathcal{S}_{n,k}$ are correctly defined, and
	The set  $\{\mathcal{S}_{n,k} : k \in \llbracket 1;n \rrbracket \}$ is a partition of $\{1,2\}^n$. Indeed, the length of the word $w_n$  is at least $n$, 
	since at each step, after reading $t_i$, we write at least one letter of $\mathcal{O}_T$. And each letter is 1 or 2, so $\sum^{|w_n|}_{i = 1} { [w_n]_i} \geq n$. 
	
	Let us notice that $ \mathcal{S}_{n,n} = \{(1,...,1), (1,...,1,2)\}$. Indeed, the length of $w_{n-1}$ is at least $n-1$, and is exactly $n-1$ if and only if the written letters are all 1's. Moreover $\mathcal{S}_{n,n-1} = \{(1,...,1,2,1),(1,...,1,2,2)\}$. Indeed, one easily check that $(1,...,2,1),(1,...,2,2) \in \mathcal{S}_{n,n-1}$. Reciprocally, if $(t_1,\dots,t_n) \in \mathcal{S}_{n,n-1}$, then there exists $i \in \llbracket 1 ; n-1 \rrbracket$ such that $t_i = 2$. Let us note that the first $2$ in $T$ is written twice. Thus, if there were such a $i$ in $\llbracket 1 ; n-2 \rrbracket$ we would have $\sum^{|w_{n-2}|}_{i = 1} { [w_{n-2}]_i} \geq n$.

\subsection{Extension of the definition to infinite directing sequences}
	
Now that we have started studying the notions of directing and directed sequences, a natural question arises: << Can one extend the definition of $\mathcal{O}_T$ sequences to infinite (possibly not periodic) sequences $T$? >> 
	Let $\mathcal{A} \subseteq \bbN^\star$ be an alphabet and let $T=(t_n)_{n \in \bbN^\star}$ be an infinite sequence over $\mathcal{A}$.
	By construction, $\mathcal{O}_{t_1\dots t_n}$ is a prefix of $\mathcal{O}_{t_1\dots t_{n+1}}$for each $n \in \bbN^\star$. One thus defines $\mathcal{O}_T$ as the limit of $\mathcal{O}_{t_1\dots t_n}$ when $n$ tends to infinity:   $\mathcal{O}_T = \displaystyle \lim_{n \to \infty } \mathcal{O}_{t_1\dots t_n}$. 
	
	The present work deals with the densities of letters in $\mathcal{O}_T$ when $T=(t_n)_{n \in \bbN^\star}$ is an infinite sequence over $\mathcal{A}$. Do these densities exist? If so, how much are they value?
	
	We are especially interested in the case where the directing sequence is random. Let $\mathbb{T} = (T_n)_{n \in \bbN^\star}$ be a sequence of random variables. By definition, the sequence directed by $\mathbb{T}$ is the random sequence $\mathbb{X} = (X_n)_{n \in \bbN^\star}$ defined by $\mathbb{X}=\mathcal{O}_{{\mathbb T}}$.
	
	The present paper is organized as follows. In section \ref{sec::i.i.d.}, we consider the case where the directing sequence is made of independent and identically distributed (i.i.d.) random variables over a two-letter alphabet. In section \ref{sec::markov}, we treat the case where the random sequence $\mathbb{X}$ is directed by a Markov chain.

%% file: Section3.tex
\section{Sequence directed by independent random variables\label{sec::i.i.d.}}

	In the present section, $\mathbb{T} = (T_n)_{n \in \bbN^\star}$ is a sequence of independent and identically distributed random variables (i.i.d. for short) over the two-letter alphabet $\mathcal{A} = \{1,2\}$, with $\mathbb{P}(T_n = 1)=p$ and $\mathbb{P}(T_n = 2)=1-p$ for each $n \in \bbN^\star$, where $p$ is a fixed parameter in $]0,1[$. The sequence $\mathbb{T}$ is thus distributed according to the \emph{product distribution} $(p\delta_1+(1-p)\delta_2)^{\otimes {\N^\star}}$. We denote $\T \sim ((p\delta_1+(1-p)\delta_2)^{\otimes {\N^\star}}$. 

Let $\mathbb{X} = (X_n)_{n \in \bbN^\star} $ be the sequence directed by $\mathbb{T}$. The sequence $\mathbb{X}$ is a random sequence with \textit{a priori} unknown distribution. Assume that one wants to compute the $n$\textsuperscript{th} letter $X_n$ of $\mathbb X$, for some large integer $n$. Unless the sequence $\mathbb T$ begins with a long succession of $1$'s (an event which has a low probability to occur), one just has to read the first terms $T_1,\ldots,T_k$ of $\mathbb T$, until knowing the length and position of the block containing $X_n$, and to fill that block by $1$'s with probability $p$, or by $2$'s with probability $1-p$. The resulting value of $X_n$ obtained that way will have the desired distribution. This leads to the fact that $\lim_{n\to\infty} {\mathbb P}(X_n=1)=p$, and we can even use this observation to compute more precisely the value of ${\mathbb P}(X_n=1)$, as shown in the following proposition.
	
\begin{proposition}\label{prop:proba}
If $\T \sim (p\delta_1+(1-p)\delta_2)^{\otimes {\N^\star}}$ with $p\in]0,1[$, then for any $n \geq 2$, 
$$\displaystyle \mathbb{P}{(X_n = 1)} = p(1-p^{n-2}+p^{n-1}).$$
\end{proposition}
	
\begin{proof}

%%%%%%%%%%%%%%%%%%%%%%%%%%%%%%%
%%%%%%%%%%%%%%%%%%%%%%%%%%%%%%%
%%%%%%%%%%%%%%%%%%%%%%%%%%%%%%%

We will decompose the event $\{X_n=1\}$ according to the partition $\{\mathcal{S}_{n,k} : k \in \llbracket 1;n \rrbracket \}$ of $\{1,2\}^n$ introduced in the previous section, and use the observations~\eqref{eq:S1} and~\eqref{eq:S2}. 
The following two particular cases are obvious:
 $$\mathbb{P}(X_n = 1 \, | \, (T_1,...,T_n) \in \mathcal{S}_{n,n}) = p, \quad \mbox{ and } \quad \mathbb{P}(X_n = 1 \, | \, (T_1,...,T_n) \in \mathcal{S}_{n,n-1}) = 0.$$
		
Let us now consider some $(T_1,...,T_n) \in \mathcal{S}_{n,k}$,  with $k<n-1$. Let $k'$ be the unique integer such that $|w_{k'-1}| < n$ and $|w_{k'}| \geq n$. Concretely, the $n$\textsuperscript{th} letter $X_n$ is written during the $k'$\textsuperscript{th} step, and it is equal to $T_{k'}$. By definition of $\mathcal{S}_{n,k}$, we have $k' \geq k$. Let us show that $k'\neq k$. We reason by contradiction and assume that $k'=k$. Let us now look at what we know at the end of the $k-1$\textsuperscript{th} step.
\begin{itemize}
\item By definition of $k'$, $X_n$ will be written in the next block to the right of $w_{k'-1}$.
\item By definition of $k$, we do not know the length of this block yet.
\end{itemize}
Consequently, we do not know the length of any (empty) block to the right of $w_{k'-1}$. It directly implies that $w_{k'-1}$ is made of 1's (and thus, that $w_{k'-1} = 1^{k'-1}$). It follows that $n \leq |w_{k'}| \leq |w_{k'-1}| + 2 = k'+1$. Since $k'=k < n-1$, we get a contradiction, which means that $k<k'$.
		
The fact that $(T_1,...,T_n)$ belongs to $\mathcal{S}_{n,k}$ only depends on the beginning $(T_1,...,T_k)$ of the sequence. Recall that $\T$ has a product distribution $(p\delta_1+(1-p)\delta_2)^{\otimes {\N^\star}}.$ Since $k'>k$, and $X_n=T_{k'}$, we deduce that
$$\mathbb{P}(X_n = 1 \, | \, (T_1,...,T_n) \in \mathcal{S}_{n,k}) = p.$$

Finally, we use the formula of total probability:
\begin{align*}
\mathbb{P}(X_n = 1)  & = \sum_{k=1}^{n-2} \; \mathbb{P}((T_1,...,T_n) \in \mathcal{S}_{n,k})\times \mathbb{P}(X_n = 1 \, | \, (T_1,...,T_n) \in \mathcal{S}_{n,k}) \\
& \qquad + \mathbb{P}((T_1,...T_{n-1}) = (1,...,1))\times \mathbb{P}(X_n = 1\, | \, (T_1,...,T_{n-1}) = (1,...,1)) \\
& \qquad + \mathbb{P}((T_1,...T_{n-1}) = (1,...,2))\times \mathbb{P}(X_n = 1 \, | \, (T_1,...,T_{n-1}) = (1,...,2)) \\
& = \big(1- (p^n + 2(1-p)p^{n-1} + p^{n-2} (1-p)^2)\big) \times p + p^n \times 1 + 0 \\
& = p(1-p^{n-2}+p^{n-1}).
\end{align*}   
\end{proof}

As a corollary, we obtain the following convergence of the proportion of $1$'s in $\X$.
	
\begin{corollary} If $\T\sim (p\delta_1+(1-p)\delta_2)^{\otimes {\N^\star}}$ with $p\in]0,1[$, then
\begin{equation*}
\lim_{n \to \infty} \dfrac{\mathbb{E}(|X_1\dots X_n|_1)}{n} = p
\end{equation*}
\end{corollary}
	
\begin{proof}
Since $\X$ has values in $\{1,2\}$, we have $|X_1\dots X_n|_1=2n-(X_1 + ... + X_n ).$ It follows that
$$\dfrac{\mathbb{E}(|X_1\dots X_n|_1)}{n}=2-\dfrac{\sum_{k=1}^n \E(X_k)}{n}.$$
By Proposition~\ref{prop:proba}, we have $\lim_{n\to\infty} \P(X_n=1)=p$ and $\lim_{n\to\infty} \P(X_n=2)=1-p$. We deduce that $\lim_{n\to\infty} \E(X_n)=2-p$.
By Ces\`aro lemma, we obtain:
$$\lim_{n\to\infty} \dfrac{\mathbb{E}(|X_1\dots X_n|_1)}{n}=2-(2-p)=p.$$
\end{proof}

Each time we run a simulation with $\T\sim (p\delta_1+(1-p)\delta_2)^{\otimes {\N^\star}}$, the frequency of $1$'s in $\X$ seems to converge to $p$. We thus expect the sequence ${|X_1\dots X_n|_1}/{n}$ to converge almost surely to $p$, and not only in expectation. Since the random variables $(X_n)_{n\in\N^\star}$ are correlated, we can not directly apply the strong law of large numbers (SLLN) to prove the almost sure convergence of $(X_1 + ... + X_n)/n$. However, the correlations being sufficiently weak, we can apply the following stronger version of the SLLN.
	
\begin{theorem}[Lyons \cite{Lyons88}]\label{thm:Lyons}
Let $(Y_n)_{n \in \mathbb{N^\star}}$ be a sequence of real-valued random variables such that for all $n \in \mathbb{N^\star},$ $|Y_n| \leq 1$ and
$$\forall n,m \in \mathbb{N^\star}, \ \mathbb{E}(Y_mY_n) \leq \Phi (|n-m|), \quad \mbox{ with } \Phi \geq 0 \mbox{ and } \sum_{n\geq 1} \frac{\Phi(n)}{n} < \infty.$$
Then $\displaystyle \lim_{n\to\infty} \frac{1}{n}\sum_{k=1}^{n} Y_k =0$ almost surely.
\end{theorem}

In order to apply Theorem~\ref{thm:Lyons}, let us first prove the following lemma.

\begin{lemma}\label{lemma:corr} If $\T\sim (p\delta_1+(1-p)\delta_2)^{\otimes {\N^\star}}$ with $p\in]0,1[$, then for any $m \geq 1$ and any $n \geq m+2$, 
$$\mathbb{P}(X_m = 2 \mbox{ and } X_{n} = 1) = p \times \mathbb{P}(X_m = 2).$$
\end{lemma}
	
\begin{proof} Since $\mathcal{S}_{n,n} \cap (X_m = 2) = \emptyset$ and $\mathcal{S}_{n,n-1} \cap (X_m = 2) = \emptyset$, we have
$$\mathbb{P}(X_m = 2 \, \cap \, X_{n} = 1) = \sum_{k=1}^{n-2} \mathbb{P}(X_m = 2 \, \cap \, X_{n} = 1 \, \cap \, (T_1,...,T_{n}) \in \mathcal{S}_{n,k}).$$
It follows that
\begin{align*}
\mathbb{P}(X_m = 2 \, \cap \, X_{n} = 1)
& = \sum_{k=1}^{n-2} \mathbb{P}((T_1,...,T_{n}) \in \mathcal{S}_{n,k}) \times \mathbb{P}(X_{n} = 1 \ | \ (T_1,...,T_{n}) \in \mathcal{S}_{n,k})\\
& \qquad \qquad \times \mathbb{P}(X_m = 2 \ |\ X_{n} = 1 \, \cap \, (T_1,...,T_{n}) \in \mathcal{S}_{n,k}) 
\end{align*}
First, observe that for $k\leq n-2$, $\mathbb{P}(X_{n} = 1 \ | \ (T_1,...,T_{n}) \in \mathcal{S}_{n,k})=p$.
Let us now prove that $$\mathbb{P}(X_m = 2 \ |\ X_{n} = 1 \, \cap \, (T_1,...,T_{n}) \in \mathcal{S}_{n,k}) = \mathbb{P}(X_m = 2 \ |\ (T_1,...,T_{n}) \in \mathcal{S}_{n,k}).$$
It is equivalent to proving that when $X_m = 2$ and $(T_1,...,T_n) \in \mathcal{S}_{n,k}$ are not incompatible,
$$\mathbb{P}(X_{n} = 1 \ |\ X_m = 2\, \cap \, (T_1,...,T_{n}) \in \mathcal{S}_{n,k}) = \mathbb{P}(X_{n} = 1 \ |\ (T_1,...,T_{n})\in \mathcal{S}_{n,k}).$$
Let $i$ be the integer such that the letter $X_m$ is given by $T_i$. 
We can decompose the event $(T_1,...,T_{n}) \in \mathcal{S}_{n,k}$ into the two following cases.\\
\begin{enumerate} 
\item If $i>k$, then after reading $(T_1,\ldots,T_k)$, we know the size of the blocks containing $X_m$ and $X_n$ but not their content:
$$\underset{\text{$w_k$}}{\ub{$\times \times \times \times \times \times \times \times $}} \ \ub{??} \ \ub{?} \ \underset{X_m}{\ub{2}} \ \ub{??} \ \ub{??} \ \underset{ X_{n}}{\ub{?}}.$$ 
In this case, the variables giving the values of $X_m$ and $X_{n}$ are independent, thus the additional information that $X_m=2$ does not affect the probability of having $X_n=1$.\\
\item If $i\leq k$, then reading $(T_1,\ldots,T_k)$ already tells us whether $X_m=2$, but does not give us the content of the block containing $X_n$, which is drawn independently: 
$$\underset{\text{$w_k$ contains $X_m$}}{\ub{$\times \times \times \times \times \times 2 \ \times $}} \ \ub{??} \ \ub{?} \ \ub{?} \ \ub{??} \ \ub{??} \ \underset{ X_{n}}{\ub{?}}.$$
\end{enumerate}
In all cases, we have 
$$\mathbb{P}(X_{n} = 1 \ |\ X_m = 2 \, \cap \, (T_1,...,T_{n}) \in \mathcal{S}_{n,k}) = \mathbb{P}(X_{n} = 1 \ |\ (T_1,...,T_{n}) \in \mathcal{S}_{n,k}) =p.$$
We deduce that
\begin{align*}
\mathbb{P}(X_m = 2 \, \cap \, X_{n} = 1)
& = \sum_{k=1}^{n-2} \mathbb{P}(T_1,...,T_{n}) \in \mathcal{S}_{n,k}) \times p \times \mathbb{P}(X_m = 2 \ |\  (T_1,...,T_{n}) \in \mathcal{S}_{n,k}) \\
& = p \times \mathbb{P}(X_m = 2).
\end{align*} 
\end{proof}
	
We can now state the following theorem.

\begin{theorem}\label{thm::thm1}
If $\T\sim (p\delta_1+(1-p)\delta_2)^{\otimes {\N^\star}}$ with $p\in]0,1[$, then
\begin{equation*}
\lim_{n \to \infty} \dfrac{ |X_1 \dots X_n|_1 } {n} = p \text{ almost surely.}
\end{equation*}
\end{theorem}
	
\begin{proof} In order to apply Theorem \ref{thm:Lyons}, we need to center the random variables $(X_n)_{n\in\N^\star}$. For $n\in\N^\star$, we thus introduce the random variables $\tilde{X}_n = X_n-(2-p),$ in order to have $|\tilde X_n|\leq 1$ and $\lim_{n\to\infty} \E(\tilde X_n)=0$. Now, let us exploit Lemma~\ref{lemma:corr} to compute $\mathbb{E}(\tilde{X}_m\tilde{X}_n)$, for $n\geq m+2$. We have
\begin{align*}
& \mathbb{P}(\tilde{X}_m = p \, \cap \, \tilde{X}_n = p-1) = p \times \mathbb{P}(X_m = 2),\\
& \mathbb{P}(\tilde{X}_m = p \, \cap \, \tilde{X}_n = p) = (1-p) \times \mathbb{P}(X_m = 2),\\
& \mathbb{P}(\tilde{X}_m = p-1 \, \cap \, \tilde{X}_n = p-1) = 1-\P(X_n=2) -p \,\mathbb{P}(X_m = 2),\\
& \mathbb{P}(\tilde{X}_m = p-1 \, \cap \, \tilde{X}_n = p) = \P(X_n=2)-(1-p)\,\P(X_m=2).
\end{align*}
Gathering these values and using Proposition~\ref{prop:proba}, we obtain
$$\E(\tilde{X}_n\tilde{X}_m)=-(1-p)^2p^{n-1}\leq 0.$$
Let us define a function $\Phi:\N^\star\to{\mathbb R}$ by $\Phi (0) = \Phi (1) = 1$ and for all $k \geq 2, \Phi (k) = 0$. Then $\mathbb{E}(\tilde{X}_m\tilde{X}_n)\leq \Phi(|n-m|)$ for all $m,n\in\N^\star$, and $\Phi$ satisfies obviously
$\Phi \geq 0$ and $\sum_{n\geq 1} \frac{\Phi (n)}{n} < \infty$.
By Theorem~\ref{thm:Lyons}, we deduce that
$$\lim_{n\to\infty}\frac{1}{n}\sum_{k=1}^{n} \tilde{X}_k = 0 \mbox{ almost surely}.$$
Consequently, $\displaystyle \lim_{n\to\infty}\frac{1}{n}\sum_{k=1}^{n} X_k= 2-p \mbox{ a.s.},$ and
$$\lim_{n \to \infty} \dfrac{ |X_1 \dots X_n|_1 } {n} = p \text{ almost surely.}$$
\end{proof}

To conclude on the case of a directing sequence following a product distribution, let us mention that the previous results can be extended to other alphabets. In particular, Proposition~\ref{prop:proba} is extended as follows.

\begin{proposition}
Let $a, b \in \bbN^{\star}$ with $1<a<b$, and let $p \in ]0,1[$.
\begin{enumerate}
\item If $\mathcal{A} = \{1,a\}$ and $\T\sim (p\delta_1+(1-p)\delta_a)^{\otimes {\N^\star}}$, then
\begin{equation*}
\forall n \geq a,\quad \mathbb{P}(X_n = 1) = p \left(1 - p^{n-a} + p^{n-1}\right)
\end{equation*}
\item If  $\mathcal{A} = \{a,b\}$ and $\T\sim (p\delta_a+(1-p)\delta_b)^{\otimes {\N^\star}}$, then
\begin{equation*}
\forall n \geq b+1, \quad \mathbb{P}(X_n = a) = p.
\end{equation*}
\end{enumerate}
\end{proposition}
	
\begin{proof}
\begin{enumerate}
\item We use the same partition as in the proof of Proposition \ref{prop:proba}, but we now distinguish the sets $\mathcal{S}_{n,k}$ for $n-a+1\leq k\leq n$. We have $\mathcal{S}_{n,n} = \{(t_1,\dots,t_n) \in \{1,a\}^n :(t_1,\dots,t_{n-1}) = (1,\dots,1)\},$ and for $n-a+1\leq k\leq n-1$,
$$\mathcal{S}_{n,k} = \{(t_1,\dots,t_n) \in \{1,a\}^n : (t_1,\dots,t_k) = (1,\dots,1,a)\}.$$
If $n-a+1\leq k\leq n-1$, then for the same reason as before, $\mathbb{P}(X_n = 1 \, | \, (T_1,\dots,T_n) \in \mathcal{S}_{n,k}) = 0$. In all the other cases, $\mathbb{P}(X_n = 1 \, | \, (T_1,\dots,T_n) \in \mathcal{S}_{n,k}) = p$. Thus, 
$$\mathbb{P}(X_n = 1) = p(1-p^{n-2}(1-p)-\dots-p^{n-a}(1-p)) = p \left(1 - p^{n-a} + p^{n-1}\right).$$
\item Since $a>1$, we have $\sum_{j=1}^{|w_k|}[w_k]_j - |w_k| \geq (a-1)|w_k| > 0$ for all $k \in \mathbb{N}^\star$. This means that we always know the length of at least one empty block after the $k$\textsuperscript{th} step. Thus, except if $n\leq b$ (in which case the $n$\textsuperscript{th} letter might be written during the first step), we are sure that we will know the length of the block containing the $n$\textsuperscript{th} letter strictly before filling it. As the $T_j$ are independent, we deduce that $\mathbb{P}(X_n = a) = p$.
\end{enumerate}
\end{proof}

%% file: Section4.tex
\section{Sequence directed by a Markov chain\label{sec::markov}}

In order to get closer to the deterministic case where a 1 always follows a 2 and vice versa, we are now interested in the case of directing sequences which are given by Markov chains.

In the present section, we assume that the directing sequence $\mathbb{T} = (T_n)_{n \in \bbN^\star}$ is a Markov chain over the alphabet $\{1,2\}$ with initial value $T_1=1$ and whose transition probability from 1 to 2 (and from 2 to 1) is $p \in ]0,1[$. A large value of $p$ encourages the alternation between 1 and 2. The original case of the Oldenburger-Kolakoski sequence can be viewed as a << limit >> case of a Markov chain whose transition probability from 1 to 2 (and from 2 to 1) would be equal to 1.

\begin{theorem}\label{thm:markov}
	Let $ p \in ]0,1[$ and let $\mathbb{T}$ be a Markov chain over the alphabet  $\{1,2\}$ with initial value $T_1=1$ and whose transition probability from 1 to 2 (and from 2 to 1) is $p \in ]0,1[$. Then 
	$$\displaystyle \lim_{n\to \infty}{\mathbb{P}(X_n = 1)} = \dfrac{1}{2}.$$
\end{theorem}

\begin{proof} Let us first note that for all integers $s>r \geq 1$, $\mathbb{P}(T_s=1 \ | \ T_r=1) = \frac{1}{2}(1+(1-2p)^{s-r})$ and $\mathbb{P}(T_s=1 \ | \ T_r=2) = \frac{1}{2}(1-(1-2p)^{s-r})$. 
	
	Let $\ell \in \mathbb{N}^\star$ and let $n \geq 2\ell$. Consider the integer $k \in \mathbb{N}^\star$ such that $(T_1,...,T_n) \in \mathcal{S}_{n,k}$.
	If we have at least $2\ell$ occurences of $2$ among $T_1,\dots,T_{\lfloor n/2 \rfloor}$, then we have at least $2\ell$ occurences of 2 among $(T_1,\dots,T_k)$, and thus at least $2\ell$ occurences of 2 in $w_k$. This implies that $n-|w_k|\geq 2\ell,$ meaning that when $w_k$ is written, we know the lengths of at least $\ell$ empty blocks between position $|w_k|$ and position $n$.
	Similarly to the proof of Proposition~\ref{prop:proba}, let us consider the integer $k'$ such that $X_n$ is given by $T_{k'}$. We also introduce $D=k-k'$, as illustrated below (by definition, $X_{|w_k|}=T_k$ and $X_n=T_{k'}$):
	$$\ub{$\times \times \times \times \times \times \times \times $} \ \underset{X_{|w_k|}}{\ub{$\times$}} \ \underbrace{\ub{??} \ \ub{?} \  \ub{??} \ \ub{??}}_{\text{$D-1$ blocks}} \ \underset{X_n}{\ub{?}}.$$
	By the above observations, we have
	\begin{align*}
	\mathbb{P}(D < \ell) &\leq \P(\{\mbox{ less that $2\ell$ occurences of $2$ among $T_1,\dots,T_{\lfloor n/2 \rfloor}$ }\}),
	\end{align*}
	and the probability on the right goes to $0$ as $n$ goes to $\infty$.
	Furthermore,		
	\begin{eqnarray*}
		\mathbb{P}(X_n=1 \ | \ D \geq \ell)  & = & \mathbb{P}(T_{k'}=1 \ | \ D \geq \ell)\\
		& = 	& \mathbb{P}(T_{k'}=1 \ | \ D \geq \ell \, \cap \, T_{k}=1)\times \mathbb{P}(T_{k} = 1 \ | \ D \geq \ell)\\
		& & \quad + \; \mathbb{P}(T_{k'}=1 \ | \ D \geq \ell \, \cap \, T_{k}=2)\times \mathbb{P}(T_{k} = 2 \ | \ D \geq \ell)\\
		& & \in \left[ \frac{1}{2}(1-|1-2p|^\ell) ; \; \frac{1}{2}(1+|1-2p|^\ell) \right] 
	\end{eqnarray*}
	thanks to the remark made at the beginning of the proof. We deduce that
	$$\underset{n \to \infty}{\limsup} \; \mathbb{P}(X_n = 1) \leq \frac{1}{2}(1+|1-2p|^\ell) \quad \mbox{ and } \quad \underset{n \to \infty}{\liminf} \; \mathbb{P}(X_n = 1) \geq \frac{1}{2}(1-|1-2p|^\ell).$$
	Then, by letting $\ell$ goes to infinity, we obtain $\lim_{n\to\infty} \P(X_n=1)=\frac{1}{2}$.
\end{proof}

As a direct consequence of Theorem~\ref{thm:markov}, we obtain the following result.

\begin{corollary}\label{cor:markov} Let $ p \in ]0,1[$ and let $\mathbb{T}$ be a Markov chain over the alphabet  $\{1,2\}$ with initial value $T_1=1$ and whose transition probability from 1 to 2 (and from 2 to 1) is $p \in ]0,1[$. Then 
	$$\lim_{n \to \infty} \dfrac{\mathbb{E}(|X_1\dots X_n|_1)}{n} = \dfrac{1}{2}$$
\end{corollary}

We conjecture that the convergence also holds almost surely but we have been unable to prove it so far, as the computation of the correlations is much more intricate in the markovian case.

\begin{conjecture} Let $ p \in ]0,1[$ and let $\mathbb{T}$ be a Markov chain over the alphabet  $\{1,2\}$ with initial value $T_1=1$ and whose transition probability from 1 to 2 (and from 2 to 1) is $p \in ]0,1[$. Then 
	\begin{equation*}
	\lim_{n \to \infty} \dfrac{ |X_1 \dots X_n|_1 } {n} = \dfrac{1}{2} \text{ almost surely.}
	\end{equation*}
\end{conjecture}

Observe that Theorem~\ref{thm:markov} and Corollary~\ref{cor:markov} easily extend to other alphabets. In particular, one obtain an identical result over the alphabet $\{1,3\}$: if $\T$ is a Markov chain with transition probability $0< p<1$ from 1 to 3 (and from 3 to 1), then the average density of $1$'s is $1/2$. 

\begin{theorem}\label{thm:markov2}
	Let $a\geq 2$ be an integer, let $ p \in ]0,1[$ and let $\mathbb{T}$ be a Markov chain over the alphabet  $\{1,a\}$ with initial value $T_1=1$ and whose transition probability from $1$ to $a$ (and from $a$ to $1$) is $p \in ]0,1[$. Then 
	$$\displaystyle \lim_{n\to \infty}{\mathbb{P}(X_n = 1)} = \dfrac{1}{2} \text{ and } \lim_{n \to \infty} \dfrac{\mathbb{E}(|X_1\dots X_n|_1)}{n} = \dfrac{1}{2}.$$
\end{theorem}
The statement of Theorem \ref{thm:markov2} is somewhat surprising and unexpected since we know that the densities $d_1$ and $d_3$ of the letters $1$ and $3$ in the sequences $\mathcal{O}_{1,3}$ and $\mathcal{O}_{3,1}$ are respectively $d_1 \approx 0,40$ and $d_3 \approx 0,60$. \cite{Sing04}. We will come back to this in the discussion of Section~\ref{sec:discussion}.

%% file: Section5.tex
	\section{Non conservation of the density}
	
	In previous sections, we have studied different cases where the directing sequences are random. In all the cases we considered (sequences of independent and identically distributed random variables, Markovian sequences), the densities of letters of the directed sequence obtained are the same as those in the directing sequence, almost surely. 
	
	Simulations also suggest that for any (infinite) periodic sequence $T$, the density of $1$'s in directed sequence $\mathcal{O}_T$ is well-defined and is equal to the density of $1$'s in $T$, see Figure~\ref{fig:periodic}. 
	
	\begin{figure}[!h]
		\begin{center}
			\includegraphics[width=.5\textwidth]{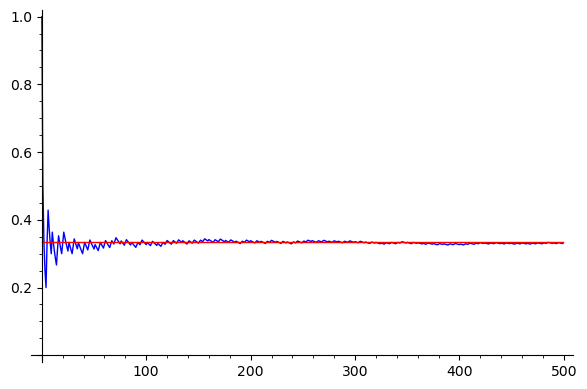}\includegraphics[width=.5\textwidth]{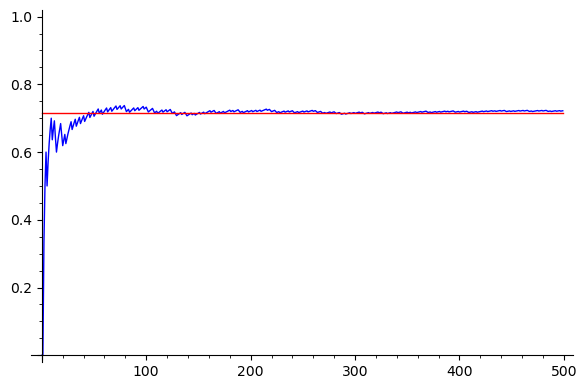}
		\end{center}
		\caption{Evolution of the density of $1$'s in increasingly large prefixes of $\mathcal{O}_T$ for $T=(122)^{\omega}$ (left) and $T=(2112111)^{\omega}$ (right). The densities seem to converge respectively to $1/3$ and to $5/7$.}
		\label{fig:periodic}
	\end{figure}

	On Figure~\ref{fig:periodic}, we have chosen to represent only the data on short prefixes of $\mathcal{O}_T$  so that it remains usable, especially to distinguish the densities in the very first terms of the sequence $ \mathcal{O}_T$. However, further experiments have been carried out on a large number of periodic sequences $T$ and they seem to corroborate our first impression, namely that if the sequence $T$ is periodic then the densities in $\mathcal{O}_T$ would be the same as those in $T$. This leads us to state the following conjecture, that extends Keane's conjecture.
	\begin{conjecture}\label{conj:conj2} For any periodic sequence $T$ over the alphabet $\{1,2\}$, the density of $1$'s in the directed sequence $\mathcal{O}_T$ is well-defined and is equal to the density of $1$'s in $T$.
	\end{conjecture}
	
	Then, a natural question arises: does there exist a directing sequence $T$ over $\{1,2\}$ for which the density of $1$'s in $T$ is not conserved in $\mathcal{O}_T$? 
	Obviously, because of Conjecture \ref{conj:conj2}, we do not expect to find such a candidate of directing sequence among the periodic ones.
	
	However, we answer this question partially and positively thanks to the fact that the left-to-right reading of $\mathcal{O}_T$ provides the size of the blocks even further to the right (see Section \ref{sec:section2}).
	In a prospect of building step by step both sequences $T$ and $\mathcal{O}_T$, the knowledge of the length of not yet filled blocks of $\mathcal{O}_T$ could allow us to choose, in a fully arbitrary way, with which letter we will fill them and it could give us the opportunity to force the sequence $\mathcal{O}_T$ to contain relatively more 1's than the sequence $T$. 
	
	The main idea of our simultaneous construction scheme of $T$ and $\mathcal{O}_T$ can be summarized as follows: we initialize $T_1$ to $2$, then $\mathcal{O}_T=22\,\ub{??}$ and from now on, by reading $\mathcal{O}_T$ from left to right, we fill its blocks of size 2 with 1's and its blocks of size 1 with 1's and 2's alternatively.	The first steps in the simultaneous construction of $T$ and $\mathcal{O}_T$ are thus as follows (with the notation of Section \ref{sec:section2}):
	\begin{enumerate}[label={\textbf{Step \arabic*:}}]
		\item We set $\mathrm{T}^{(1)}=(2)$ and then $\mathcal{O}^{(1)} = \ub{22}\,\ub{??}$% \, \ldots$
		\item 
			\begin{enumerate}
				\item The empty block of $\mathcal{O}^{(1)}$ of size $2$ must be filled with $1$'s: $\mathcal{O}^{(2)} = \ub{22}\,\ub{11} \, \ub{?} \, \ub{?}$ 
				\item We set $\mathrm{T}^{(2)}=(2,1)$
			\end{enumerate}
		\item 
			\begin{enumerate}
				\item We fill the next block of $\mathcal{O}^{(2)}_T$ of size $1$ with $1$ : $\mathcal{O}^{(3)}=\ub{22}\,\ub{11}\,\ub{1}\,\ub{?} \, \ub{?}$
				\item Hence $\mathrm{T}^{(3)}=(2,1,1)$%
			\end{enumerate}
		\item 			
			\begin{enumerate}
				\item We fill the next block of $\mathcal{O}^{(3)}_T$ of size $1$ with $2$ : $\mathcal{O}^{(4)}=\ub{22}\,\ub{11}\,\ub{1}\,\ub{2} \, \ub{?} \, \ub{??}$
				\item Hence $\mathrm{T}^{(4)}=(2,1,1,2)$%
		\end{enumerate}
		\item and son on\dots
	\end{enumerate}

	For each $n \in \mathbb{N}^\star$, we have $\mathcal{O}^{(n)} =  \mathcal{O}_{\mathrm{T}^{(n)}}$.  Moreover, $\mathrm{T}^{(n)}$ is a prefix of $\mathrm{T}^{(n+1)}$ while $\mathcal{O}^{(n)}$ is a prefix of $\mathcal{O}^{(n+1)}$, then we set $T = \lim_{n \to \infty}{\mathrm{T}^{(n)}}$. It follows that $\mathcal{O}_T = \lim_{n\to\infty}{\mathcal{O}^{(n)}}$.

	Let us denote $T = (t_i)_{i \in \mathbb{N}}$ and  $\mathcal{O}_T = (x_i)_{i \in \mathbb{N}}$ with $t_i, x_i \in \{1,2\}$ for all $i \in \mathbb{N}$, then:
	\begin{enumerate}
		\item $\mathrm{T}^{(n)} = (t_1,\dots,t_n)$ and $|\mathrm{T}^{(n)}|=n$. 
		\item $\mathcal{O}^{(n)} = (t_i^{x_i})_{i \in [\![1,n]\!]}$ and $\mathcal{O}_T = (t_i^{x_i})_{i \in \mathbb{N^\star}}$.
		\item $|\mathrm{T}^{(n)}|_1  = |x_1 \dots x_n|_2 - 1 + \dfrac{1}{2}|x_1 \dots x_n|_1 + C_n$, with $C_n \in \{0,1\}$: indeed, the number of $1$'s in $\mathrm{T}^{(n)}$ is equal to the sum of the number of blocks of size $2$ in $\mathcal{O}^{(n)}$ (except the first block of $\mathcal{O}^{(n)}$ because of the initialisation of $\mathcal{O}^{(1)}$) and half of the number of blocks of size $1$ in $\mathcal{O}^{(n)}$. 
		 By construction, the number of blocks of size $1$ (resp. of size $2$) in $\mathcal{O}^{(n)}$ is equal to the number of $1$'s (resp. $2$'s) in $x_1\dots x_n$. The constant $C_n$ takes into account the cases where $x_n=1$  and is the first letter of a block of size 2 in $\mathcal{O}^{(n)}$.
	\end{enumerate}
	Program~\ref{lst::liste2} provides a \texttt{Python} function for the construction of $\mathcal{O}_{T}$ and $T$.
	
	\begin{minipage}{\linewidth}
		\begin{lstlisting}[caption={\texttt{Python} function for the simultaneous construction of $T$ and $\mathcal{O}_T$.},language=Python, label={lst::liste2}]
	def Sequences(n) :
		T = [2]
		O_T = [2, 2]
		d = 1 # digit to write in the next block of size 1
		for i in range(1, n) :
			if O_T[i] == 2 :
				T += [1]
				O_T += [1]*2
			else : 
				T += [d]
				O_T += [d]
				d = 3-d
		return (T, O_T)
	\end{lstlisting}
	\end{minipage}

\begin{theorem}
	Let $T = \lim_{n \to \infty}{\mathrm{T}^{(n)}}$. The following properties hold: 
	\begin{enumerate}
	%	\item The density $d^T_1$ of $1$'s in $T$ exists if and only if so does the density $d^{\O}_1$ of $1$ in $\mathcal{O}_T$.
		\item If the density $d^T_1$ of $1$'s in $T$ exists, then 
		$d^T_1 = {1+\sqrt{17}\over 8}=0.640\ldots$, $d_1^{\O}={7-\sqrt{17}\over 4}=0.719\ldots$ and so $d^T_1 \neq d^{\O}_1$.
		\item If the density $d^T_1$ of $1$'s in $T$ exists, then the sequences $T$ and $\mathcal{O}_T$ are not periodic.
	\end{enumerate}
\end{theorem}

\begin{proof}
	\begin{enumerate}
		\item 
	For each $n \in \mathbb{N}^\star$, we have:
	\begin{equation*}
	\Big | \; 
	|\mathcal{O}^{(n)}|_1 - |{\mathrm{T}^{(n)}}|_2-2(|{\mathrm{T}^{(n)}}|_1-|{\mathrm{T}^{(n)}}|_2) \; \Big |\leq 1 %& \text{ and }
	%|\mathcal{O}^{(n)}|_2= |{\mathrm{T}^{(n)}}|_2+1.
	%& 
	\end{equation*}
Indeed, to within one unit, each digit $2$ of ${\mathrm{T}^{(n)}}$ gives rise to a single $2$ in $\mathcal{O}^{(n)}$, and a same quantity $|{\mathrm{T}^{(n)}}|_2$ of $1$'s gives rise to a single $1$ in $\mathcal{O}^{(n)}$, while the rest of them (so $|{\mathrm{T}^{(n)}}|_1-|{\mathrm{T}^{(n)}}|_2$) give rise to two $1$'s in $\mathcal{O}^{(n)}$. 
	Moreover the first $2$ of ${\mathrm{T}^{(n)}}$ is the only one to be written twice in $\mathcal{O}^{(n)}$, so that we always have exactly $
	|\mathcal{O}^{(n)}|_2= |{\mathrm{T}^{(n)}}|_2+1$. Then, 
%	Let 
%	$$d^T_1 = \lim_{n \to \infty}{\dfrac{|T^{(n)}|_1}{|T^{(n)}|_1+|T^{(n)}|_2}}$$ be the density of $1$'s in the sequence $T$.
	\begin{eqnarray*}
	\dfrac{|\mathcal{O}^{(n)}|_1}{|\mathcal{O}^{(n)}|} & = & 	\dfrac{|\mathcal{O}^{(n)}|_1}{|\mathcal{O}^{(n)}|_1 + |\mathcal{O}^{(n)}|_2} 
	= 		\dfrac{2|T^{(n)}|_1 - |T^{(n)}|_2 +o(n)}{-1+|T^{(n)}|_2 + 2(|T^{(n)}|_1- |T^{(n)}|_2) + |T^{(n)}|_2+1    +o(n)) } \\
	& = & 	  \dfrac{3|T^{(n)}|_1 - |T^{(n)}| + o(n)}{2|T^{(n)}|_1 +o(n)} \underset{n \to \infty}\longrightarrow \dfrac{3d^T_1-1}{2d^T_1}
	\end{eqnarray*}	
We conclude that 
\begin{equation}\label{eq:eq5}
d^{\O}_1  = 	\dfrac{3d^T_1-1}{2d^T_1}
\end{equation} and the density of $1$'s (resp. of $2$'s) in $\mathcal{O}_T$ exists.	

We noticed above that, for each $n \in \mathbb{N}^\star$, $|\mathrm{T}^{(n)}|_1  = |x_1 \dots x_n|_2 - 1 + \dfrac{1}{2}|x_1 \dots x_n|_1 + C_n$, with $C_n \in \{0,1\}$. Moreover, if $d^T_1$ exists then so do $d^O_1$ and $d^O_2$ and, by tending $n$ towards infinity, we have:
\begin{equation}\label{eq:eq6}
d^T_1 = d^O_2+\dfrac{1}{2}d^O_1
\end{equation}
	By putting together equations \eqref{eq:eq5} and \eqref{eq:eq6}, we deduce	$d_1^T={1+\sqrt{17}\over 8}$ and  $d_1^{\O}={7-\sqrt{17}\over 4}$.
\item If the sequences $T$ and $\mathcal{O}_T$ were periodic, then their densities of $1$'s and $2$'s would be rational, which is not the case.
	\end{enumerate}

\end{proof}
Simulations suggest that the densities are indeed converging to these values, see Figure~\ref{fig:WT}.
	\begin{figure}
		\begin{center}
			\includegraphics[width=.6\textwidth]{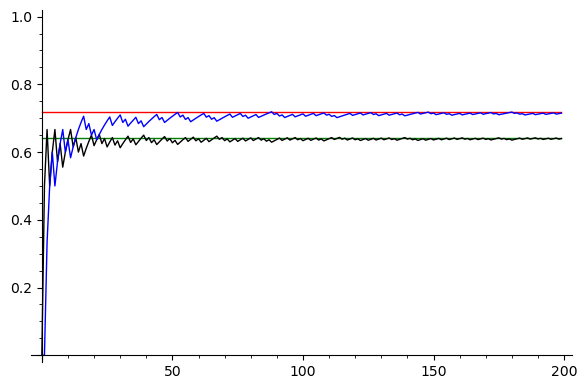}
		\end{center}
		\caption{Evolution of the densities of $1$'s in $\O_T$ (blue) and $T$ (black), where the two sequences are defined by Program~\ref{lst::liste2}.}
		\label{fig:WT}
	\end{figure}

%% file: Section6.tex
	\section{Conclusion and discussion}\label{sec:discussion}
	Over the alphabet $\{1,a\}$, with $a \in \{2,3\}$, we have shown that in almost all the sequences directed by an infinite sequence  $\mathbb{T} = (T_n)_{n \in \bbN^\star}$  of i.i.d. random variables with $\mathbb{P}(T_n = 1) = p \in ]0,1[$ and $\mathbb{P}(T_n = a) = 1-p$, the density of $1$'s is equal to $p$.
	We have also shown that the average density of $1$'s among all sequences directed by a Markov chain with transition probability $p \in ]0,1[$  from $1$ to $a$ and from $a$ to 1 is equal to $1/2$.

	Keane's conjecture \cite{Keane91} states that this result can be extended to the deterministic case, namely when $p=1$, over the alphabet $\{1,2\}$. On the other hand, over the alphabet $\{1,3\}$ this result is not extendable  to the deterministic case since the density of $1$'s in $\mathcal{O}_{1,3}$ is close to 0.3972. \cite{Sing04}. 
	
	When $\mathbb{T}$ is a Markov chain, the closer its transition probability $p$ is to $1$, the more likely the sequence $\mathcal{O}_T$ is to share a long prefix with $\mathcal{O}_{1,3}$. Therefore, the closer the transition probability $p$ is to 1, the closer the density of 1's in the sequence $\mathcal{O}_T$ is to that in the sequence $\mathcal{O}_{1,3}$ on a long prefix. However, computer experiments suggest that when the first perturbations in the alternation of 1's and 3's appear in $\mathbb{T}$, the density of $1$'s in the prefix of $\mathcal{O}_\mathbb{T}$ eventually approaches 0.5 as this prefix gets longer. See Figure~\ref{fig:M13} for an illustration with $p=0.99$. 
	
	\begin{figure}[!h]
		\begin{center}
			\includegraphics[width=.6\textwidth]{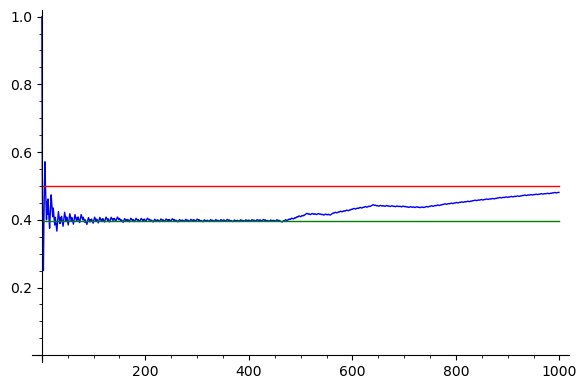}
		\end{center}
		\caption{Evolution of the frequency of $1$'s for a markovian directing sequence on the alphabet $\{1,3\}$ of parameter $p=0.99$: the frequency is first close to the one of $\mathcal{O}_{1,3}$ then moves away from it to converge to $1/2$.}
		\label{fig:M13}
	\end{figure}

	This implies it seems difficult to derive information about the original Oldenburger-Kolakoski sequence $\mathcal{O}_{1,2}$ by letting $p$ tend to $1$ in the Markovian case over the alphabet $\{1,2\}$.
	
	Finally, the study of sequences directed by random sequences on alphabets of more than 2 letters or by random sequences constructed from other distributions also seems interesting.

%% file: BoissonJametMarcovici.bbl
\begin{thebibliography}{10}

\bibitem{Sing04}
Michael Baake and Bernd Sing.
\newblock Kolakoski-(3, 1) is a (deformed) model set.
\newblock {\em Canadian Mathematical Bulletin}, 47(2):168–190, 2004.

\bibitem{BBC05}
Val{\'{e}}rie Berth{\'{e}}, Srecko Brlek, and Philippe Choquette.
\newblock Smooth words over arbitrary alphabets.
\newblock {\em Theor. Comput. Sci.}, 341(1-3):293--310, 2005.

\bibitem{BertheBC05}
Val{\'{e}}rie Berth{\'{e}}, Srecko Brlek, and Philippe Choquette.
\newblock Smooth words over arbitrary alphabets.
\newblock {\em Theor. Comput. Sci.}, 341(1-3):293--310, 2005.

\bibitem{BDLV06}
Srecko Brlek, Serge Dulucq, A.~Ladouceur, and Laurent Vuillon.
\newblock Combinatorial properties of smooth infinite words.
\newblock {\em Theor. Comput. Sci.}, 352(1-3):306--317, 2006.

\bibitem{BJP08}
Srecko Brlek, Damien Jamet, and Genevi{\`{e}}ve Paquin.
\newblock Smooth words on 2-letter alphabets having same parity.
\newblock {\em Theor. Comput. Sci.}, 393(1-3):166--181, 2008.

\bibitem{Carpi93}
Arturo Carpi.
\newblock Repetitions in the {K}olakovski sequence.
\newblock {\em Bull. {EATCS}}, 50:194--197, 1993.

\bibitem{chvatal93}
Va\v{s}ek Chv\'atal.
\newblock Notes on the {K}olakoski sequence.
\newblock Technical report, DIMACS Technical Report 93-84, December 1993.

\bibitem{Dek80}
F.~M. Dekking.
\newblock On the structure of selfgenerating sequences.
\newblock {\em Séminaire de Théorie des Nombres de Bordeaux}, pages 1--6,
  1980.

\bibitem{Keane91}
Michael~S. Keane.
\newblock Ergodic theory and subshifts of finite type.
\newblock In {\em Ergodic theory, symbolic dynamics, and hyperbolic spaces.
  Lectures given at the workshop "Hyperbolic geometry and ergodic theory", held
  at the International Centre for Theoretical Physics in Trieste, Italy, 17-28
  April, 1989}, pages 35--70. Oxford etc.: Oxford University Press, 1991.

\bibitem{WK1966}
William Kolakoski.
\newblock Self-generating runs, problem 5304.
\newblock {\em The American Mathematical Monthly}, 73(6):681--682, 1966.

\bibitem{Lyons88}
Russell Lyons.
\newblock {Strong laws of large numbers for weakly correlated random
  variables.}
\newblock {\em Michigan Mathematical Journal}, 35(3):353 -- 359, 1988.

\bibitem{RO1939}
Rufus Oldenburger.
\newblock Exponent trajectories in symbolic dynamics.
\newblock {\em Transactions of the American Mathematical Society},
  46(3):453--466, 1939.

\bibitem{Sing11}
Bernd Sing.
\newblock More {K}olakoski sequences.
\newblock {\em Integers}, 11B:A14, 2011.

\end{thebibliography}
